\documentclass{article}

\usepackage{Filipovic_Larsson_Ware}
\numberwithin{table}{section}

\title{Polynomial processes for power prices%
\thanks{This work was supported by the Natural Sciences and Engineering Research Council of Canada (NSERC).}
}
\author{Damir Filipovic%
\thanks{EPFL and Swiss Finance Institute, Extranef 218, CH-1015 Lausanne, Switzerland, \texttt{damir.filipovic@epfl.ch}}%
\and  Martin Larsson%
\thanks{Department of Mathematics, ETH Z\"urich, R\"amistrasse 101, CH-8092 Z\"urich, \texttt{martin.larsson@math.ethz.ch}}
\and Tony Ware%
\thanks{Department of Mathematics and Statistics, University of Calgary, 2500 University Drive NW, Calgary, Canada. T2N 1N4, \texttt{aware@ucalgary.ca}}}
\date{\today}

\begin{document}
\maketitle

\begin{abstract}
Polynomial processes have the property that expectations of polynomial functions (of degree $n$, say) of the future state of the process conditional on the current state are given by polynomials (of degree $\leq n$) of the current state. Here we explore the application of polynomial processes in the context of structural models for energy prices. We focus on the example of Alberta power prices, derive one- and two-factor models for spot prices. We examine their performance in numerical experiments, and demonstrate that the richness of the dynamics they are able to generate makes them well suited for modelling even extreme examples of energy price behaviour.
\end{abstract}

\section{Introduction}
The class of polynomial processes is characterized by the property that the expectation of any polynomial function (perhaps up to some degree $n$, say) of the future state of the process, conditional on its current state, is given by a polynomial (of no higher degree)) of the current state. It includes exponential L\'evy processes, affine processes, as well as Pearson diffusions, allowing for non-trivial dynamics on compact state spaces, for example.
The use of polynomial processes in financial modeling goes back at least to the early 2000s, with the work of Delbaen and Shirikawa 
\cite{DelbaenShirakawa2002} and Zhou 
\cite{Zhou2003} on interest rate modelling. More recently, %
Cuchiero et.~al.~\cite{CuchieroKeller-ResselTeichmann2012} have given a systematic treatment of time-homogeneous Markov jump-diffusion polynomial processes, along with applications to variance reduction techniques for option valuation and hedging.

Filipovic and Larsson 
\cite{FilipovicLarsson2016} lay out the mathematical foundations for polynomial diffusions on a large class of state spaces, and describe their potential advantages in a wide range of applications in finance, such as market models for interest rates, credit risk, stochastic volatility and energy commodities. It is this last area that is the focus of this paper.

Energy market prices are notable for exhibiting seasonality, mean reversion and extreme volatility, often in the form of short-lived `spikes'. These features arise from the interaction between supply and demand, with each subject to seasonal variations, and in particular from constraints on storage, transmission or transportation, perhaps exacerbated by sudden changes due to unforeseen infrastructure breakdowns. (For recent overviews of modelling and risk management issues in energy markets, we refer to \cite{BenthBenthKoekebakker2008,Kaminski2012,CarmonaCoulon2014,BurgerGraeberSchindlmayr2014,Swindle2015}.)

Power prices are perhaps an extreme example, due in no small part to the difficulties in storing electricity efficiently in any significant volumes (although advances in battery technology are beginning to change this), which means that the relationship between supply and demand in power markets plays a significant role in price formation, especially in deregulated markets. Here we focus on one such market in Alberta, Canada, which deregulated in the late 1990s.

\subsection{The Alberta power market}
The Alberta wholesale real-time electricity market is facilitated by the Alberta Electric Systems Operator (AESO). The AESO system controller sets a System Marginal Price (SMP) in response to system demand, and according the prevailing merit order. The merit order, or \emph{bid stack} is determined from supply and demand bids received for each hour, which are sorted in order from the lowest to the highest price. At each moment, the SMP corresponds to the last eligible electricity block dispatched by the system controller. At the end of each hour, a time-weighted average of the SMPs is published as the pool price. The market in Alberta has a price cap of CAD\$1000, and a price floor of \$0. Producers are obliged to offer all of their available capacity into the market, but are under no obligation to ensure that prices correspond to variable cost. 

An example bid stack, for noon to 1pm on 7th March, 2012 (obtained from \url{aeso.ca}), is shown in Figure~\ref{fig:bidstack}. Bids ranged from \$0 to \$999.99. On average, 7713MW were dispatched through the hour, and the pool price settled at \$22.67. The shape of this curve varies from hour to hour within the day to reflect the anticipated demand. The steep rise in the curve above (in this case) about 8000MW means that prices can easily rise dramatically. This potential is reflected in Figure~\ref{fig:bidstack}, which shows daily average prices from 1998 to early 2016. 

\begin{figure}
\centering
\includegraphics[width=0.4\linewidth]{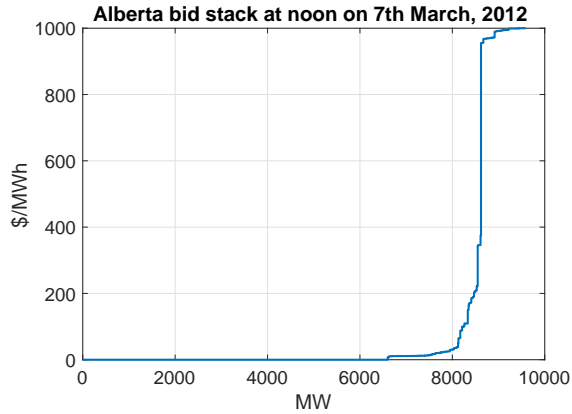}
\caption{The Alberta power pool bid stack at noon on 7th March, 2012. Note that just over 6000MW has been bid in at zero. In this hour, an average of 7713MW was dispatched, and the pool price settled at \$22.67 (source: \url{aeso.ca}). As can be seen, if demand had been such that much more than 8000MW needed to be dispatched, the SMP would have risen sharply.}
\label{fig:bidstack}
\end{figure}

\begin{figure}
\centering
\includegraphics[width=0.85\linewidth]{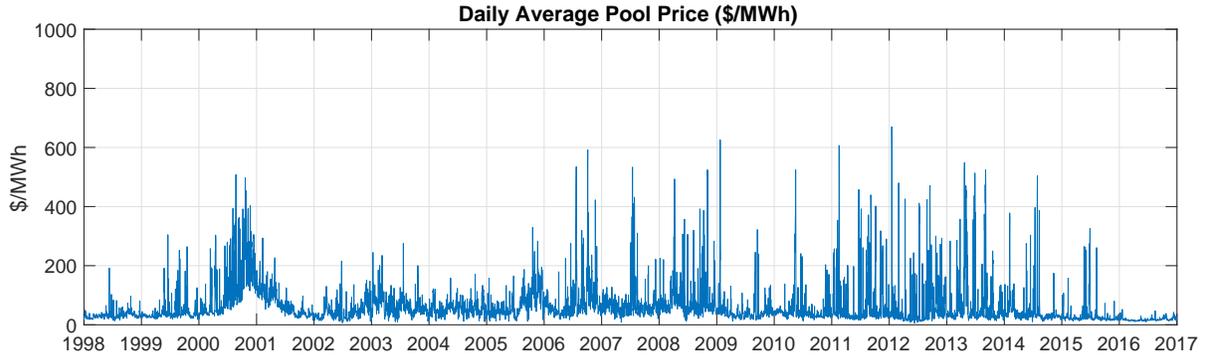}
\caption{Daily average Alberta pool prices between 1998 and 2016.}
\label{fig:dailyprices}
\end{figure}

\subsection{Structural models for power prices}
Structural models attempt to capture aspects of the interaction between supply and demand, while maintaining some level of mathematical tractability (c.f.~\cite{Pirrong2011,CarmonaCoulon2014,Weron2014}). 
One of the earliest such models for power prices is that of Barlow \cite{Barlow2002}. He takes demand to be a linear function of a factor $X_t$, modelled as an Ornstein-Uhlenbeck (OU) process:
\begin{equation}
\label{eq:OU}
\dd X_t = -\kappa(X_t-\theta)\dd t + \sigma \dd W_t,
\end{equation} 
where $W_t$ is a standard Brownian motion.
Spot prices, denoted by $S_t$, are given (for some $\alpha<0$) by
\[ S_t = f_{\alpha}(X_t) = \begin{cases}\big(1+\alpha X_t\big)^{\frac{1}{\alpha}} &\word{if $1+\alpha X_t>\epsilon_0$,}\\ \epsilon_0^{\frac{1}{\alpha}}&\word{otherwise,}
\end{cases}
\]
where $\epsilon_0$ is chosen to reflect the maximum price in the market.

\begin{figure}
\centering
\includegraphics[width=0.9\linewidth]{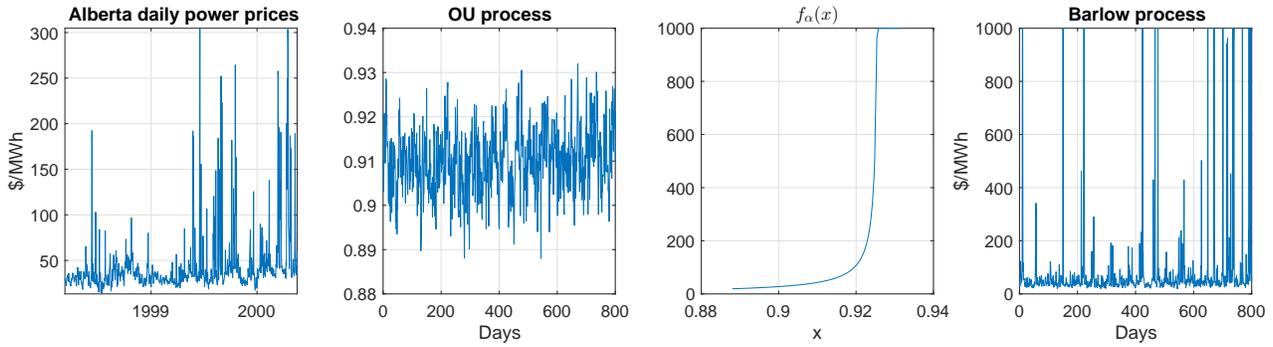}
\caption{Barlow's structural model for power prices \cite{Barlow2002}. From left to right: historical daily average prices for the period 10th March 1998-18th May 2000; a simulation of $X_t$ produced using calibrated parameters for that period taken from Table~4.1 of \cite{Barlow2002}; the function $f_{\alpha}(x)$; the resulting prices $S_t = f_{\alpha}(X_t)$.}
\label{fig:barlow}
\end{figure}

As can be seen in Figure~\ref{fig:barlow}, even though the underlying demand process is a diffusion process, the model is capable of generating the kinds of short-lived spikes in prices that are evident in the historical market price time series, although the simulated spikes are if anything even more extreme than the historical values, and the \$1000 limit is frequently reached. The presence of this cutoff value means that it is not possible to obtain explicit expressions for forward prices within the framework of this model, something that Barlow himself notes (\cite{Barlow2002}).

Nevertheless, Barlow's model has had a strong influence on the subsequent development of structural models. Kanamura and \=Ohashi 
\cite{KanamuraOhashi2007}, and Boogert and Dupont 
\cite{BoogertDupont2008} take a similar approach, but use a piecewise polynomial map in place of the Box-Cox transform. Some (for example \cite{CarteaVillaplana2008,ElliottLyle2009}) have generalized the approach by introducing additional underlying factors such as total market capacity, also following a mean-reverting diffusion process, and expressed prices as an exponential function of a linear combination of these factors. This generates models with a similar mathematical structure to that of Schwartz and Smith 
\cite{SchwartzSmith2000} for capturing short- and long-term dynamics in commodity prices. 

Other authors have expressed the power prices as a `heat-rate' multiplier of a power of a marginal fuel price, where the heat rate is some function of the underlying demand factors \cite{PirrongJermakyan2008}. There may be more than one such multiplier, corresponding to different market conditions. For example, Coulon et.~al.~\cite{CoulonPowellSircar2013} generate spot prices in the form
\[ S_t = (-1)^{\overline{\delta}_i}\Psi_i(G_t,L_t),\]
where the choice of function $\Psi_i$ is made with probability $\pi_i$ depending on $L_t$ and $\overline{\delta}_i\in\{0,1\}$ ($\overline{\delta}_i=1$ corresponds to negative prices), and takes the form
\[ \Psi_i(g,l) = g^{\delta_i}\ee^{\alpha_i+\beta_i\max(0,\min(l,C_{\max}))},\]
with $\delta_i\in\{0,1\}$ indicating the presence of fuel price dependence.

Dependence on multiple fuel prices can also be captured in the structural modelling paradigm. A framework for capturing dependence on several fuel prices involving a stochastic bid stack function is developed in \cite{HowisonCoulon2009,CarmonaCoulonSchwarz2013}. Their framework allows for random changes in the merit order, and---at least for two underlying fuels---can generate explicit formulae for forward prices. A\"id et.~al.~\cite{AidCampiLangrene2013} assume one bid price per fuel type, and also are able to capture spikes by incorporating a heat rate function involving a power law of the reserve margin. They argue that the use of the power law (instead of an exponential map) makes it possible to reproduce sharp spikes even for smooth and rather simple dynamics of the underlying demand and capacity processes.

\section{Polynomial processes for energy markets}
The main idea we present in this paper is that energy market models can be generated  by using a polynomial map to capture the role of the heat rate (or bid stack) function, and using a polynomial process for the underlying factor(s). As we hope to demonstrate, this provides a framework that allows for more general relationship between underlying factors and resulting prices than that afforded by the use of exponential or power law maps, while ensuring tractable forward price formulae.

If $X_t$ follows a polynomial process, then we can generate a spot price model in the form
\begin{equation}
\label{eq:spotpolynomial}
S_t = \Phi(X_t) = H(X_t)^\top  \vec p,
\end{equation}
where $H(x)$ is a vector of basis functions for the space of polynomials preserved by the polynomial process (for example, in one dimension, $H(x) = (1,x,x^2,\ldots,x^n)^\top$ for some fixed $n\in\N$), and $\vec p$ is the coefficient vector that defines the polynomial map $\Phi$.

Seasonality may be incorporated by making the polynomial coefficient vector $\vec p$ time-dependent. Setting $\vec p(t) = \sum_{k=1}^K s_k(t) \vec p_k$ for some fixed vectors $\vec p_1,\dots,\vec p_K$ and deterministic functions $s_1(t),\ldots,s_K(t)$ yields
\begin{equation}
\label{eq:spotpolynomialS}
S_t = H(X_t)^\top \sum_{k=1}^K s_k(t) \vec p_k.
\end{equation}

Moreover, we can write futures prices straightforwardly, exploiting the polynomial property of the factor process. We have
\begin{equation}
\label{eq:forwardpolynomial}
F(t,T,T') = \frac{1}{T'-T} H(X_t)^\top \sum_{k=1}^K \left[\int_T^{T'} e^{(u-t)G} s_k(u)du\right]\vec p_k,
\end{equation}
where the interval $[T,T']$ is the contract delivery period, and
$G$ is the matrix representation of the generator of $X_t$ with respect to the basis $H(x)$ under the pricing measure $\Q$. That is, if $\Phi(x)=H(x)^\top\vec p$ is a polynomial of degree at most $n$, then $\E_{\Q}[\Phi(X_t)|X_0]=H(X_0)^\top e^{tG}\,\vec p$.

\begin{example}
A tractable example of seasonality weight is $s(t)=\cos(ct)$ where $c$ is a constant. Indeed, since $\cos(ct)={\rm Re}(e^{{\rm i} ct})$ we have the explicit expression
\begin{align*}
\int_T^{T'} e^{(u-t)G} s(u)du
&= {\rm Re}\left( \int_T^{T'} e^{(u-t)G} e^{{\rm i} cu} du \right) = {\rm Re}\left( e^{{\rm i}ct} \int_T^{T'} e^{(u-t)(G+{\rm i}c)} du \right) \\
&= {\rm Re}\left( e^{{\rm i}ct} (G+{\rm i}c)^{-1} \left( e^{(T'-t)(G+{\rm i}c)} - e^{(T-t)(G+{\rm i}c)}\right) \right),
\end{align*}
provided $G+{\rm i}c$ is invertible, i.e.~$-{\rm i}c$ is not an eigenvalue of $G$. In particular, this certainly holds if $G$ has only real eigenvalues. Clearly, a seasonality weight expressed as a linear combination of such Fourier modes can be dealt with in a similar manner.
\end{example}

\subsection{One factor specifications}
The class of one-factor polynomial diffusion processes includes Geometric Brownian Motion, the OU process \eqref{eq:OU}, as well as other mean-reverting processes such as Inhomogeneous Geometric Brownian Motion (IGBM) \cite{Pilipovic1997,BosWarePavlov2002}, defined by
\begin{equation}
\label{eq:IGBM}
\dd X_t = \kappa(\theta-X_t) \dd t + \sigma X_t\dd W_t,
\end{equation}
where $\kappa>0$, $\theta> 0$, $\sigma>0$, and, here and below, $W$ is Brownian motion on some filtered probability space $(\Omega,\Fcal,\Fcal_t,\P)$. Whereas for the OU process the state space is $\R$, for IGBM \eqref{eq:IGBM} $X_t\in (0,\infty)$ for $t>0$ a.s.~if $X_0>0$. 

The Cox-Ingersoll-Ross process is also polynomial:
\begin{equation}
\label{eq:CIR}
\dd X_t = \kappa(\theta-X_t)\dd t + \sigma \sqrt{X_t}\dd W_t.
\end{equation}
For the CIR process \eqref{eq:CIR}, $X_t\in(0,\infty)$ a.s.~if $2\kappa\theta\geq \sigma^2$.

Our last example in this section is the Jacobi process, which will form the foundation for the models we explore further below. It is defined by
\begin{equation}
\label{eq:Jacobi}
\dd X_t = \kappa(\theta - X_t)\,\dd t + \sigma\sqrt{X_t(1-X_t)}\,\dd W_t,
\end{equation}
where $\kappa>0$, $\theta\in[0,1]$, and $\sigma>0$. Here the state space is $[0,1]$, and 
$X_t\in (0,1)$ a.s.~if $2\kappa\min\{\theta,1-\theta\geq \sigma^2 $ (see Appendix A for more details). 

\subsection{Some two-factor specifications with bounded state spaces}

A simple extension of the Jacobi one-factor model described above is the following:
\begin{align*}
\dd X_t &= (b_1 + B_{11}X_t + B_{12}Y_t) \dd t + \sigma\sqrt{X_t(1-X_t)}\dd W_{1t} \\
\dd Y_t &= (b_2 + B_{22}Y_t) \dd t + \rho \sqrt{Y_t(1-Y_t)}\dd W_{2t}.
\end{align*}
Here $Y_t$ acts as a level factor that affects the level of mean reversion of $X_t$. In this simple specification, $Y_t$ is an autonomous Jacobi process, and there is no quadratic covariation between the factors: $\langle X,Y\rangle=0$. Other possible dynamics for $Y_t$ are discussed below.

In this setting the spot price and futures prices are given by almost identical formulas; the only difference is the appearance of the second factor $Y_t$:
\begin{align*}
S_t &= H(X_t,Y_t)^\top \sum_{k=1}^K s_k(t) \vec p_k \\
F(t,T,T') &= \frac{1}{T'-T} H(X_t,Y_t)^\top \sum_{k=1}^K \int_T^{T'} e^{(u-t)G} s_k(u)\dd u\,\vec p_k,
\end{align*}
where:
\begin{itemize}
\item $H(x) = (1,x,y,x^2,xy,y^2,\ldots,x^n,x^{n-1}y,\ldots,xy^{n-1},y^n)^\top$ for some fixed $n\in\N$.
\item $s_1(t),\ldots,s_K(t)$ are deterministic functions capturing the seasonality in prices.
\item $\vec p_k$ are the coordinate representations of the polynomials that map $(X_t,Y_t)$ to the spot price~$S_t$. 
\item The interval $[T,T']$ is the delivery period of the underlying.
\item $G$ is the matrix representation of the generator of $(X_t,Y_t)$ with respect to the basis $H(x,y)$. That is, if $\Phi(x,y)=H(x,y)^\top\vec p$ is a polynomial of degree at most $n$, then $\E[p(X_t,Y_t)]=H(X_0,Y_0)^\top e^{tG}\,\vec p$.
\end{itemize}

Some other possible polynomial preserving factor specifications are as follows:

\begin{example}
\label{eg:DoubleJacobi}
Feedback from $X_t$ into the drift of $Y_t$:
\begin{align*}
\dd X_t &= (b_1 + B_{11}X_t + B_{12}Y_t) \dd t + \sigma\sqrt{X_t(1-X_t)}\dd W_{1t} \\
\dd Y_t &= (b_2 + B_{21}X_t + B_{22}Y_t) \dd t + \rho\sqrt{Y_t(1-Y_t)}\dd W_{2t}.
\end{align*}
\end{example}

\begin{example}
The range of $X_t$ depending on $Y_t$:
\begin{align*}
\dd X_t &= (b_1 + B_{11}X_t + B_{12}Y_t) \dd t + \sigma\sqrt{X_t(\mu+\nu Y_t-X_t)}\dd W_{1t} \\
\dd Y_t &= (b_2 + B_{21}X_t + B_{22}Y_t) \dd t + \rho\sqrt{Y_t(1-Y_t)}\dd W_{2t}.
\end{align*}
for suitable parameters $\mu\ge0$ and $\nu\ge0$. Here $X_t$ takes values in $[0,\mu+\nu Y_t]$ and $Y_t$ takes values in $[0,1]$. Thus the state space is $E=\{(x,y): 0\le x\le \mu+\nu y, 0\le y\le 1\}$, which is a parallelogram.
\end{example}

\begin{example}
In all the above examples we have $\langle X,Y\rangle=0$. The following specification relaxes this condition:
\begin{align*}
\dd X_t &= (b_1 + B_{11}X_t + B_{12}Y_t) \dd t + \sqrt{1-X_t^2-Y_t^2}\, ( \alpha_{11} \dd W_{1t} + \alpha_{12} \dd W_{2t} )\\
\dd Y_t &= (b_2 + B_{21}X_t + B_{22}Y_t) \dd t + \sqrt{1-X_t^2-Y_t^2}\, ( \alpha_{12} \dd W_{1t} + \alpha_{22} \dd W_{2t} ),
\end{align*}
where $\alpha=(\alpha_{ij})_{i,j=1,2}$ is a symmetric positive definite matrix. This can be written in vectorized form as
\[
\begin{pmatrix}\dd X_t\\ \dd Y_t\end{pmatrix} = \left( b + B\begin{pmatrix}X_t\\Y_t\end{pmatrix} \right)\dd t + \sqrt{1-X_t^2-Y_t^2}\, \alpha\, \dd W_t,
\]
where $b=(b_1,b_2)^\top$, $B=(B_{ij})_{i,j=1,2}$, and $W_t=(W_{1t},W_{2t})$. The state space for this process is the unit disk, $E=\{(x,y): x^2+y^2\le 1\}$. In particular, $X_t$ takes values in $[-\sqrt{1-Y_t^2},\sqrt{1-Y_t^2}]$.
\end{example}

\begin{example}
\label{eg:RegimeSwitching}
Regime switching model:
\begin{align*}
\dd X_t &= (b_1 + B_{11}X_t + B_{12}Y_t) \dd t + \sigma\sqrt{X_t(1-X_t)}\dd W_t \\
\dd Y_t &= (1-Y_t)\dd N^0_t+Y_t\dd N^1_t,
\end{align*}
where $N^0_t$ and $N^1_t$ are standard Poisson processes with intensities $\lambda_{0\to 1}$ and $\lambda_{1\to 0}$ (resp.). With $Y_0\in\{0,1\}$, $Y_t$ alternates between the two values $0$ and $1$. Thus the state space is $E=[0,1]\times\{0,1\}$. To compute the generator of $(X,Y)$, consider a $C^2$ function $f(x,y)$ on $E$. It\^o's formula yields
\begin{align*}
f(X_t,Y_t) &= f(X_0,Y_0) + \int_0^t \left(  (b_1 + B_{11}X_s + B_{12}Y_s)\partial_x f(X_s,Y_s) + \frac{1}{2}\sigma^2 X_s(1-X_s)\partial_{xx} f(X_s,Y_s)\right)\dd s\\
&\quad + \int_0^t \partial_x f(X_s,Y_s)\sigma\sqrt{X_s(1-X_s)} \dd W_s + \sum_{s\le t}\left( f(X_s,Y_s) - f(X_s,Y_{s-}) \right) \\
&= f(X_0,Y_0) + \int_0^t \Gcal f(X_s,Y_s) \dd s + {\rm(local\ martingale)},
\end{align*}
where
\[
\Gcal f(x,y) = (b_1 + B_{11}x + B_{12}y)\partial_x f(x,y) + \frac{1}{2}\sigma^2 x(1-x)\partial_{xx} f(x,y)
+\lambda_{0\to 1}(1-y)f(x,0) -\lambda_{1\to 0}y f(x,1).
\]
For this state space there are fewer polynomials than on $\R^2$. Indeed, any polynomial $p(x,y)$ on $E$ is of the form $p(x,y) = (1-y)p_0(x)+ yp_1(x)$. Thus a simpler basis can be used in this case, for example
\[
H(x,y) = (1,x,y,x^2,xy,x^3,x^2y,x^4,x^3y,\cdots,x^n,x^{n-1}y)^\top.
\]
Moreover, if $B_{12}=0$, a tensor product basis can be used.
\end{example}

\subsection{Option valuation}
It turns out that, in certain settings at least, we can derive explicit option pricing formulae. Specifically, if the underlying factor process $X_t$ admits an explicit formula for European call or put options on powers of the factor (as is the case if $X_t$ follows a geometric Brownian motion), and $\Phi$ is an increasing polynomial map, then we can derive explicit formulae for vanilla European options written on $S_t = \Phi(X_t)$. For example, if we write the formula for the value of a European call option with payoff $(X_T^j-K^j)_+$ as $f_j(K)$, and $\Phi(x) = \sum_{j=0}^na_jx^j$, then the value of a European call option with payoff $(S_T-K)_+$ is given by
\[ \sum_{j=0}^n a_j f_j\left( \Phi^{-1}(K)\right).\]
More generally, option prices can be approximated by approximating the payoff by a polynomial of some degree, and computing the expectation under $\Q$ using the matrix representation $G$ of the generator of $X_t$ under $\Q$, and perhaps exploiting fast polynomial transforms where available \cite{PottsSteidlTasche1998,Keiner2011}. This can be seen as analogous to transform methods such as the COS method of Fang and Oosterlee~\cite{FangOosterlee2008}. Furthermore, polynomial expansion methods for option pricing in polynomial jump-diffusion models are discussed in Section 7 in \cite{FilipovicLarsson2016}.

\section{Modelling spot prices in Alberta's electricity market}
In this section we demonstrate the use of polynomial processes to model daily average spot prices in Alberta's electricity market. Similarly to Barlow~\cite{Barlow2002}, we do not see strong evidence of annual seasonality in these prices, so we do not incorporate seasonal terms in our spot price models.
\subsection{One-factor model}
Our one-factor model takes the form of \eqref{eq:spotpolynomial}, with $X_t$ following a Jacobi process \eqref{eq:Jacobi}, and the (increasing) polynomial map $\Phi(\cdot)$ generated according to the prescription laid out in Appendix~\ref{AppendixB} (but scaled so that $\Phi:[0,1]\mapsto [0,S_{\max}]$, with $S_{\max}=1000$ in the case of Alberta). The model with $\Phi$ of degree $n+1$ is thus determined by the parameter vector $\Theta = (\kappa,\theta,\sigma,c_1,\dots,c_n)$, where the coefficients $c_i$ correspond (in pairs) to the parameters $(\alpha,\beta)$ in the quadratic factors $\phi$ described in Appendix~\ref{AppendixB}, with the last value assigned to a pair $(0,\beta)$ if $n$ is odd.

The estimation of these parameters, for each $n$, was done using maximum likelihood estimation, with the optimization carried out using a combination of genetic search and pattern search from the global optimization toolbox in MATLAB. The pure-Jacobi model ($n=0$) was calibrated first, and the optimal parameters for each value of $n$ were used (after being augmented by $c_{n+1}=0$) to supply the initial estimate for the calibration for $n+1$.

Note that, for a given set of parameters, the conditional transition probability density for the underlying factor process $X_t$ observed on successive days\footnote{In the estimation we conduct here we are implicitly assuming that the daily average prices correspond to instantaneous observations of the process $S_t$.} is given by the function $p$ defined in \eqref{eq:JacobiTransition}, using $T-t=1/365$, and that the unconditional density is $w$. However, $X_t = \Phi^{-1}(S_t)$. Thus, if we write $s_m$ for the price observed on day $t_m$, $m=0,\dots,M$, and $x_m = \Phi^{-1}(s_m)$, then the log-likelihood can be written
\begin{equation}
LL = \sum_{m=0}^M \log p(x_m,t_m;x_{m-1},t_{m-1})-\log \diff{}{x}\Phi(x_m),
\end{equation}
where, for $m=0$, $p(x_m,t_m;x_{m-1},t_{m-1})$ collapses to the unconditional density $w(x_0)$.

\begin{table}
\begin{center}
\begin{tabular}{|r||l|l|l|l|l|l|l|l|l|l|l|l|}
\hline
deg $\Phi$ & $\kappa$ & $\theta$ & $\sigma$ & $a$ & $b$ & $c_1$& $c_2$& $c_3$& $c_4$& $c_5$ &$LL$ & BIC\\
\hline
1 & 284.73 &  0.05 &  2.81 &  3.25 & 68.85 & &&&&& 2037.3 & -4054.6 \\
2 & 237.93 &  0.15 &  3.16 &  7.24 & 40.32 & 0.86 &&&&& 2149.7 & -4272.6 \\
3 & 146.76 &  0.29 &  4.77 &  3.74 &  9.19 & 0.95 &   0.64&&&& 2492.2 & -4950.9 \\
4 & 143.96 &  0.37 &  4.70 &  4.81 &  8.21 & 0.86 &   0.76 &   0.73&&& 2501.5 & -4962.8 \\
5 & 140.10 &  0.42 &  6.09 &  3.18 &  4.37 & 0.17 &   0.91 &   0.96 &   0.59&& 2534.9 & -5023.0 \\
6 & 140.14 &  0.43 &  6.11 &  3.26 &  4.26 & 0.04 &   0.91 &   0.95 &   0.60 &   0.19& 2535.1 & -5016.8 \\
\hline
\end{tabular}
\caption{Optimal parameters, log-likelihoods and Bayesian Information Criterion (BIC) scores for calibration of the one-factor model for the period 10th March, 1998 to 18th May, 2000, with various degrees of polynomial map $\Phi$. The values $a$ and $b$ are as defined in \eqref{eq:ab}. Notice that $a\geq 1$ and $b\geq 1$ in all cases. The most favourable BIC score is for degree~5.}
\label{tab:Table1}
\end{center}
\end{table}

We start by performing a calibration on one of the datasets considered by Barlow~\cite{Barlow2002}. This is the period from 10th March, 1998 to 18th May, 2000, and Figure~\ref{fig:barlow} shows Barlow's model fitted to this dataset. The calibration results for $n=0,\dots,5$ are shown in Table~\ref{tab:Table1}.

\begin{figure}
\centering
\includegraphics[width=0.9\linewidth]{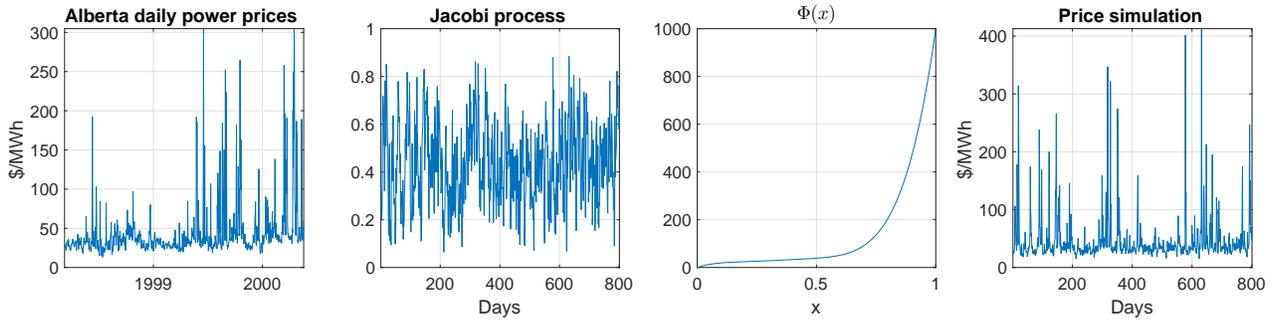}
\caption{A one-factor polynomial model fitted to Alberta prices. From left to right: historical daily average prices for the period 10th March 1998-18th May 2000; a simulation of $X_t$ produced using calibrated parameters for that period (c.f.~Table~\ref{tab:Table1}); the degree-5 polynomial map $\Phi(x)$; the resulting prices $S_t = \Phi(X_t)$.}
\label{fig:OneFactorFitA}
\end{figure}

The Bayesian Information Criterion allows us to compare the performance of models with varying numbers of parameters, and penalizes models with more parameters. As the degree of $\Phi$ is increased, the log likelihood increases, but peaks at degree~5. The resulting polynomial map and a sample simulation is shown in Figure~\ref{fig:OneFactorFitA}. It is clear that the model is successfully able to reproduce the short-term spikes evident in the historical data, and the varying slopes of the polynomial map mean that the model is able to capture varying levels of volatility at different price levels.

\begin{table}
\begin{center}
\begin{tabular}{|r||l|l|l|l|l|l|l|l|l|l|l|l|}
\hline
deg $\Phi$ & $\kappa$ & $\theta$ & $\sigma$ & $a$ & $b$ & $c_1$& $c_2$& $c_3$& $c_4$& $c_5$ &$LL$ & BIC\\
\hline
1 & 263.19 &   0.07 &   5.47 &   1.30 &  16.29 &     &     &     &     &     & 2709.6 & -5397.3 \\
2 & 210.87 &   0.19 &   5.70 &   2.50 &  10.47 &   0.85 &     &     &     &     & 2910.3 & -5791.4 \\
3 & 165.51 &   0.34 &   7.28 &   2.13 &   4.11 &   0.95 &   0.62 &     &     &     & 3405.8 & -6775.1 \\
4 & 162.62 &   0.43 &   6.94 &   2.90 &   3.85 &   0.86 &   0.76 &   0.79 &     &     & 3422.4 & -6801.2 \\
5 & 150.16 &   0.53 &   6.91 &   3.30 &   2.98 &   0.20 &   0.91 &   0.97 &   0.55 &     & 3464.7 & -6878.3 \\
6 & 150.13 &   0.53 &   6.91 &   3.32 &   2.97 &   0.18 &   0.91 &   0.97 &   0.55 &   0.03 & 3464.7 & -6871.0 \\
\hline
\end{tabular}
\caption{Optimal parameters, log-likelihoods and Bayesian Information Criterion (BIC) scores for calibration of the one-factor model for the period from 1st January, 2010 to 1st January, 2014, with various degrees of polynomial map $\Phi$. The values $a$ and $b$ are as defined in \eqref{eq:ab}. Notice that $a\geq 1$ and $b\geq 1$ in all cases. The most favourable BIC score is for degree~4.}
\label{tab:Table2}
\end{center}
\end{table}

As a second experiment, we consider the period from 1st January, 2010 to 1st January, 2014 (see Figure~\ref{fig:dailyprices}). The calibration results are shown in Table~\ref{tab:Table2}. This time the lowest BIC score is for degree~4, and the corresponding model is illustrated in Figure~\ref{fig:OneFactorFitB}.

\begin{figure}
\centering
\includegraphics[width=0.9\linewidth]{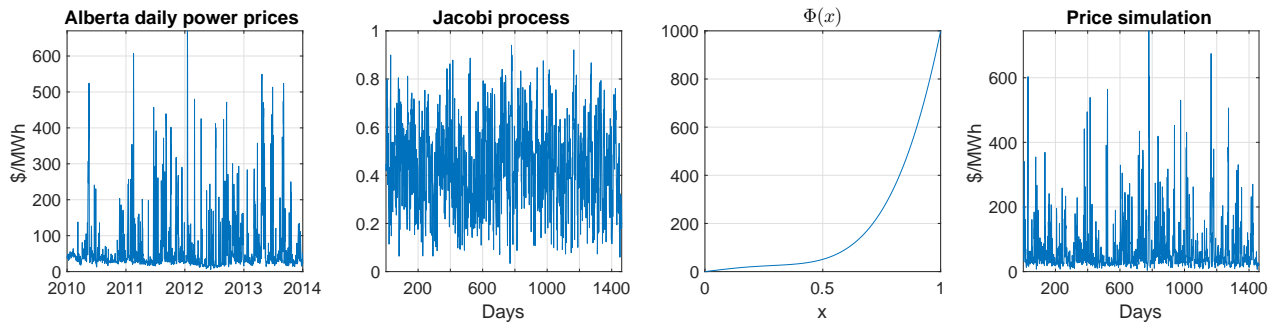}
\caption{A one-factor polynomial model fitted to Alberta prices. From left to right: historical daily average prices for the period 1st January, 2010 to 1st January, 2014; a simulation of $X_t$ produced using calibrated parameters for that period (c.f.~Table~\ref{tab:Table2}); the degree-4 polynomial map $\Phi(x)$; the resulting prices $S_t = \Phi(X_t)$.}
\label{fig:OneFactorFitB}
\end{figure}

The model is able to capture the increased, sustained level of spikes in this data set, using a slightly lower degree of polynomial map than in the first experiment. However, it is evident from both experiments that this the spikes rely on quite a highly volatile Jacobi process, with a strong level of mean reversion. This is a common feature of one-factor models for such time series, and is one reason for wanting to consider multi-factor models. This we do in the next section.

\subsection{Two-factor models}
One way to construct an $N+1$-factor model is to start with a
polynomial process on the unit simplex \cite{FilipovicLarsson2016}: 
\[Y_t\in \big\{[0,1]^N|Y_{1,t}+\dots Y_{N,t} = 1\} .\]
Armed with such a process, together with an independent Jacobi process $X_t$, we can construct a finite number of (possibly time-dependent) increasing polynomial maps $\Phi_n: [0,1]\mapsto [0,S_{\max}]$, $n=1,\dots,N$, with coefficient vectors $\vec{p}_n$, and generate prices via
\[
S_t = H(X_t)\sum_n Y_{n,t}\vec{p}_n.
\]
In the case of a two-factor model, $Y_t$ is a scalar process on $[0,1]$, and this equation can be written in the form
\begin{equation}
\label{eq:TwoFactorS}
S_t = H(X_t) \left[ (1-Y_t)\vec{p}_0 + Y_t\vec{p}_1\right] = (1-Y_t)\Phi_0(X_t) + Y_t\Phi_1(X_t).
\end{equation}

The factors $X_t$ and $Y_t$ are not directly observed, so a filtering method is needed for maximum likelihood calibration.
\subsubsection{Filtering equations}
\label{sec:Filtering}
We start with observations $s_m$ on day $t_m$, for $m=0,\dots,M$, and we wish to estimate the log likelihood
\[ LL = \sum_{m=0}^{M}\log p_{S_{m}|\vec{s}_{m-1}}(s_{m};\vec{s}_{m-1}), \]
where we have (slightly abusively) written $S_{m}$ for $S_{t_m}$, and used $\vec{s}_m=(s_m,\dots,s_0)$ for the collection of observations up to time $t_m$ and $\vec{s}_m=(S_m,\dots,S_0)$ for the corresponding vector of random variables. $p_{S_{m}|\vec{s}_{m-1}}(s;\vec{s}_{m-1})$ denotes the density of $S_{m}$ conditional on the observations $\vec{s}_{m-1}$.

The conditional transition density for $X_t$ is given by
\begin{equation} 
\label{eq:JacobiTransitionDiscrete}
p_{X_{m}|X_{m-1}}(x;x_{m-1}) = p(x,t_{m};x_{m-1},t_{m-1}),
\end{equation}
where, as in the one-factor case, $p$ is given by \eqref{eq:JacobiTransition} using $T-t=1/365$, and the parameters corresponding to those governing the Jacobi process for $X_t$. The unconditional density for $X_t$ is $p_{X_{\infty}}(x) = w(x)$.

If we assume we also know the conditional transition density $p_{Y_{m}|Y_{m-1}}(y;y_{m-1})$, and the unconditional density $p_{Y_{\infty}}(y)$, we can estimate $LL$ using optimal Bayes filtering (see \cite{Sarkka2013} for a recent treatment). This involves four steps: initialization, prediction, update and normalization, where the last three steps are carried out in a cycle, for $m=0,\dots,M$.
\begin{description}
\item[Initialization] The prior distribution for $(X_{-1},Y_{-1})$ is $p_{X_{-1},Y_{-1}|\emptyset}(x,y) = p_{X_{\infty}}(x)p_{Y_{\infty}}(y)$.
\item[Prediction] The Chapman-Kolmogorov equation yields the predicted conditional density
\begin{equation}
\label{eq:Prediction}
p_{X_{m},Y_{m}|\vec{s}_{m-1}}(x,y;\vec{s}_{m-1}) = \iint p_{X_{m}|X_{m-1}}(x;x')p_{Y_{m}|Y_{m-1}}(y;y')p_{X_{m-1},Y_{m-1}|\vec{s}_{m-1}}(x',y';\vec{s}_{m-1})\dd x' \dd y'.
\end{equation}
%\[p(x_k|y_{1:k-1}) = \int p(x_k|x_{k-1}p(x_{k-1}|y_{1:k-1})\dd x_{k-1}.\]
\item[Update] The Bayes rule yields
\begin{equation}
\label{eq:BayesRule}
p_{X_{m},Y_m|\vec{S}_m}(x,y;\vec{s}_m) = \frac{1}{Z_m}p_{S_m|X_m,Y_m}(s_m;x,y)p_{X_m,Y_m|\vec{S}_{m-1}}(x,y;\vec{s}_{m-1})
\end{equation}
%\[ p(x_k|y_{1:k}) = \frac{1}{Z_k}p(y_k|x_k)p(x_k|y_{1:k-1}).\]
\item[Normalization] The normalization constant $Z_m = p_{S_m|\vec{S}_{m-1}}(s_m;\vec{s}_{m-1})$ is given by
\begin{equation}
\label{eq:Normalization}
Z_m = \iint p_{S_m|X_m,Y_m}(s_m;x,y)p_{X_m,Y_m|\vec{S}_{m-1}}(x,y;\vec{s}_{m-1})\dd x\dd y.
\end{equation}
%\[ Z_k = \int p(y_k|x_k)p(x_k|y_{1:k-1})\dd x_k.\]
\end{description}
On completion of this process, $LL$ is given by $LL=\sum_{m=0}
^M\log Z_m$.

From \eqref{eq:TwoFactorS}, the density of $S_m$ conditional on $X_m$ and $Y_m$ is a Dirac delta distribution
\begin{equation}
\label{eq:TwoFactorDirac}
p_{S_m|X_m,Y_m}(s;x,y) = \delta_{(1-y)\Phi_0(x)+y\Phi_1(x)}(s).
\end{equation}
If, for $y\in[0,1]$ and $s$ given we define $\hat{x}(s,y)$ to be the (unique) solution of $(1-y)\Phi_0(x)+y\Phi_1(x)=s$, we can also write this as a Dirac delta distribution in $x$:
\begin{equation}
\label{eq:TwoFactorDiracx}
p_{S_m|X_m,Y_m}(s;x,y) = \frac{\delta_{\hat{x}(s,y)}(x)}{(1-y)\Phi'_0(x)+y\Phi'_1(x)}.
\end{equation}
This means that the double integrals in \eqref{eq:Prediction} and \eqref{eq:Normalization} are reduced to one-dimensional integrals over $y'$ and $y$.

\subsubsection{Regime-switching model}
Perhaps the simplest example of this setup is to take $Y_t\in\{0,1\}$ to be a continuous time Markov chain with rate parameters $\lambda_{0\to 1}$ and $\lambda_{1\to 0}$. This brings us into the setting of Example~\ref{eg:RegimeSwitching}, with  $B_{12}=0$. In this case we can write, for $y'\in\{0,1\}$, with $h$ denoting the time step between successive observations,
\begin{equation}
p_{Y_m|Y_{m-1}}(y;y') = P_{y'\to 0}(h)\delta_0(y)+P_{y'\to 1}\delta_1(y)
\end{equation}
or, alternatively, for $y\in\{0,1\}$,
\begin{equation}
p_{Y_m|Y_{m-1}}(y;y') =  P_{0\to y}\delta_0(y')+P_{1\to y}\delta_1(y'),
\end{equation}
where, for $j\in\{0,1\}$, and with $\lambda = \lambda_{0\to 1}+\lambda_{1\to 0}$,
\begin{equation}
P_{j\to 1-j} = \frac{\lambda_{j\to 1-j}}{\lambda}\big(1-\exp(-\lambda h)\big),
\end{equation}
and we have that
and $P_{j\to j} = 1-P_{j\to 1-j}$.

This regime-switching setting therefore results in a prior distribution of $p_{Y_{\infty}}(y) = \frac{\lambda_{1\to 0}}{\lambda}\delta_0(y)+ \frac{\lambda_{0\to 1}}{\lambda}\delta_1(y)$.  It also means that, using \eqref{eq:BayesRule} and \eqref{eq:TwoFactorDiracx}, \eqref{eq:Prediction} collapses to
\begin{equation}
\label{eq:RegimePrediction}
p_{X_{m},Y_{m}|\vec{s}_{m-1}}(x,y;\vec{s}_{m-1}) = \sum_{j=0}^1p_{m,j}(x)\delta_j(y)
\end{equation}
where, for $j=0,1$,
\begin{align}
\notag
p_{m,j}(x) &= P_{0\to j}\int p_{X_{m}|X_{m-1}}(x;x')p_{X_{m-1},Y_{m-1}|\vec{s}_{m-1}}(x',j;\vec{s}_{m-1})\dd x'\\
\notag &= \frac{P_{0\to j}}{Z_{m-1}\Phi'_j\big(\hat{x}(s_{m-1},j)\big)}\;p_{X_{m}|X_{m-1}}(x;\hat{x}(s_{m-1},j))\;p_{X_{m-1},Y_{m-1}|\vec{s}_{m-2}}(\hat{x}(s_{m-1},j),j;\vec{s}_{m-1})\\
&= \frac{P_{0\to j}\sum_{j'=0}^1p_{m-1,j'}(\hat{x}(s_{m-1},j'))}{Z_{m-1}\Phi'_j\big(\hat{x}(s_{m-1},j)\big)}\;p_{X_{m}|X_{m-1}}(x;\hat{x}(s_{m-1},j)).
\label{eq:Regimep}
\end{align}
Furthermore, using \eqref{eq:Normalization}, \eqref{eq:TwoFactorDiracx} and \eqref{eq:RegimePrediction}, $Z_m$ is given by
\begin{equation}
\label{eq:RegimeZ}
Z_m = \sum_{j=0}^1\frac{p_{m,j}(\hat{x}(s_{m},j))}{\Phi'_j(\hat{x}(s_{m},j))}.
\end{equation}
Finally, we note from \eqref{eq:JacobiTransition} that the transition density $p_{X_{m}|X_{m-1}}(x;x')$ and therefore the functions $p_{m,j}(x)$ are expressed as (in practice truncated) sums of Jacobi polynomials, and so \eqref{eq:Regimep} can  turned into a matrix equation between the coefficients of the expansions of $p_{m,0}$ and $p_{m,1}$ and of  $p_{m-1,0}$ and $p_{m-1,1}$.

\begin{table}
\begin{center}
\begin{tabular}{|r||r|r|r|r|r|r|}
\hline deg $\Phi_i$ & 1 & 2 & 3 & 4 & 5 & 6\\ 
\hline \hline $a$ & 3.47 & 7.55 & 5.69 & 6.34 & 7.74 & 7.34\\ 
\hline $b$ & 74.80 & 85.08 & 17.04 & 17.39 & 22.53 & 21.31\\ 
\hline $\sigma$ & 2.86 & 1.74 & 3.29 & 3.20 & 2.85 & 2.91\\ 
\hline \hline $\lambda_{0\to 1}$ & 28.44 & 19.64 & 21.60 & 24.99 & 23.75 & 25.24\\ 
\hline $\lambda_{1\to 0}$ & 11.84 & 244.98 & 217.12 & 214.94 & 207.34 & 211.07\\ 
\hline \hline $c^0_1$ & &0.60 &0.93 &0.91 &0.92 &0.92\\ 
\hline $c^0_2$ & & &0.61 &0.63 &0.64 &0.64\\ 
\hline $c^0_3$ & & & &0.19 &-0.02 &-0.13\\ 
\hline $c^0_4$ & & & & &-0.29 &-0.24\\ 
\hline $c^0_5$ & & & & & &0.11\\ 
\hline \hline $c^1_1$ & &-0.78 &1.00 &-1.00 &-1.00 &-1.00\\ 
\hline $c^1_2$ & & &-0.31 &-0.72 &-0.94 &-0.62\\ 
\hline $c^1_3$ & & & &0.88 &0.90 &-1.00\\ 
\hline $c^1_4$ & & & & &-0.39 &-0.62\\ 
\hline $c^1_5$ & & & & & &0.83\\ 
\hline \hline LL & 2041.5 & 2384.9 & 2555.8 & 2559.0 & 2559.4 & 2560.4\\ 
\hline BIC & -4049.6 & -4722.9 & -5051.4 & -5044.5 & -5031.9 & -5020.4\\ 
\hline
\end{tabular}
\caption{Optimal parameters, log-likelihoods and Bayesian Information Criterion (BIC) scores for calibration of the regime-switching two-factor model for the period from 10th March, 1998 to 18th May, 2000, with various degrees of polynomial map $\Phi$. The values $a$ and $b$ are as defined in \eqref{eq:ab}. Notice that $a\geq 1$ and $b\geq 1$ in all cases. The most favourable BIC score is for degree~3.}
\label{tab:Table3}
\end{center}
\end{table}

The results from using this procedure for our first historical data set are shown in Table~\ref{tab:Table3}. Notice that the most favourable BIC score is for $\Phi_0$ and $\Phi_1$ of degree~3, and that this score represents an improvement over the best score from the one-factor model (see Table~\ref{tab:Table1}). A simulation produced with the degree~3 parameters is shown in Figure~\ref{fig:RegimeSwitchingFitA}.

\begin{figure}
\centering
\includegraphics[width=\linewidth]{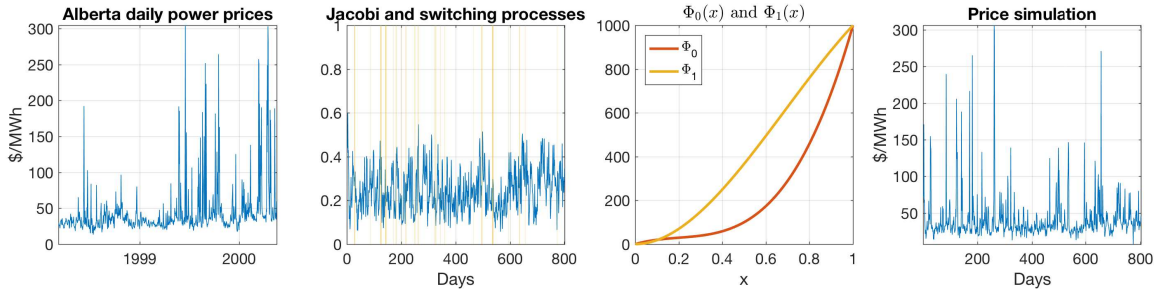}
\caption{A regime-switching two-factor polynomial model fitted to Alberta prices. From left to right: historical daily average prices for the period 10th March 1998-18th May 2000; a simulation of $X_t$and $Y_t$ produced using calibrated parameters for that period (c.f.~Table~\ref{tab:Table3}), with the times when $Y_t=1$ indicated in gold; the degree-3 polynomial maps $\Phi_0(x)$ and $\phi_1(x)$; the resulting prices $S_t = \Phi(X_t)$.}
\label{fig:RegimeSwitchingFitA}
\end{figure}
\begin{table}
\begin{center}
\begin{tabular}{|r||r|r|r|r|r|r|}
\hline deg $\Phi$ & 1 & 2 & 3 & 4 & 5 & 6\\ 
\hline \hline $a$ & 1.30 & 5.05 & 3.80 & 5.77 & 5.22 & 5.59\\ 
\hline $b$ & 16.31 & 42.15 & 10.35 & 9.70 & 5.62 & 5.38\\ 
\hline $\sigma$ & 5.47 & 3.11 & 4.62 & 4.38 & 5.03 & 5.00\\ 
\hline \hline $\lambda_{0\to 1}$ & 1.02 & 37.57 & 32.96 & 34.77 & 23.65 & 21.24\\ 
\hline $\lambda_{1\to 0}$ & 1.42 & 184.83 & 136.93 & 136.18 & 119.14 & 115.77\\ 
\hline \hline $c^0_1$ & &0.74 &0.93 &0.82 &0.23 &0.04\\ 
\hline $c^0_2$ & & &0.59 &0.73 &0.89 &0.90\\ 
\hline $c^0_3$ & & & &0.85 &0.96 &0.96\\ 
\hline $c^0_4$ & & & & &0.54 &0.55\\ 
\hline $c^0_5$ & & & & & &0.30\\ 
\hline \hline $c^1_1$ & &-1.00 &1.00 &-1.00 &-0.75 &-0.86\\ 
\hline $c^1_2$ & & &-0.93 &-0.96 &0.28 &0.50\\ 
\hline $c^1_3$ & & & &0.99 &1.00 &1.00\\ 
\hline $c^1_4$ & & & & &0.50 &0.53\\ 
\hline $c^1_5$ & & & & & &0.70\\ 
\hline \hline LL & 2711.6 & 3287.3 & 3510.7 & 3520.7 & 3546.0 & 3549.3\\ 
\hline BIC & -5386.8 & -6523.7 & -6955.7 & -6961.3 & -6997.2 & -6989.4\\ 
\hline
\end{tabular}
\caption{Optimal parameters, log-likelihoods and Bayesian Information Criterion (BIC) scores for calibration of the regime-switching two-factor model for the period from 1st January, 2010 to 1st January, 2014, with various degrees of polynomial map $\Phi$. The values $a$ and $b$ are as defined in \eqref{eq:ab}. Notice that $a\geq 1$ and $b\geq 1$ in all cases. The most favourable BIC score is for degree~3.}
\label{tab:Table4}
\end{center}
\end{table}

The results from using this procedure for the second historical data set (2010-2013) a set are shown in Table~\ref{tab:Table4}. Notice that the most favourable BIC score is for $\Phi_0$ and $\Phi_1$ of degree~5, and that this score represents a significant improvement over the best score from the one-factor model (see Table~\ref{tab:Table1}). A simulation produced with the degree~5 parameters is shown in Figure~\ref{fig:RegimeSwitchingFitB}. It is apparent for both data sets that the volatility of the Jacobi process is much reduced in comparison to the one-factor model, and that the spikes are being generated at least in part by the switching process, with the gold regions in the second graph in each case indicating the times when $Y_t=1$.

\begin{figure}
\centering
\includegraphics[width=\linewidth]{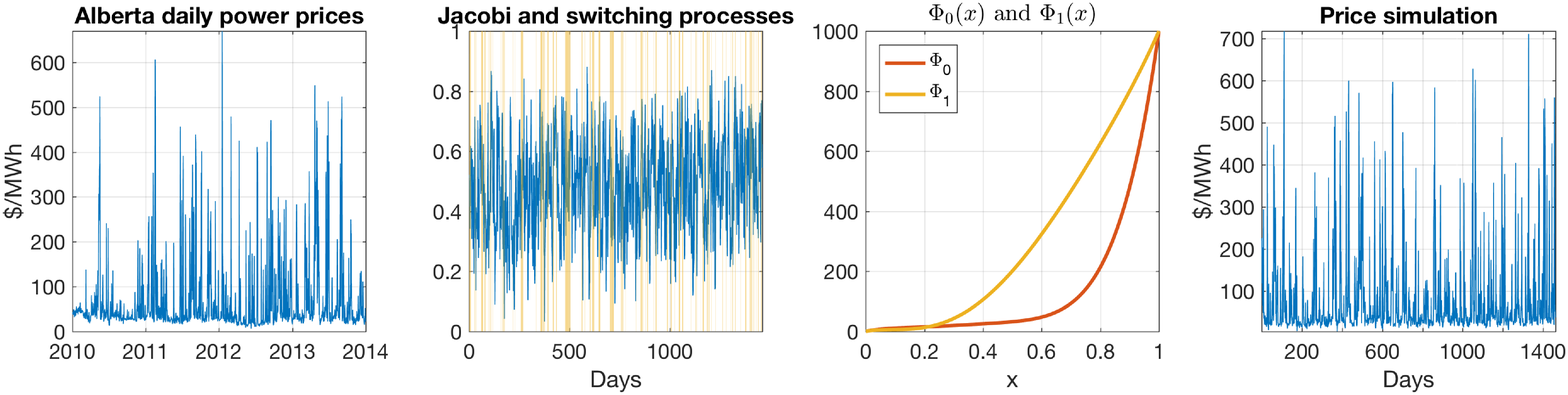}
\caption{A regime-switching two-factor polynomial model fitted to Alberta prices. From left to right: historical daily average prices for the period 1st January, 2010 to 1st January, 2014; a simulation of $X_t$ and $Y_t$ produced using calibrated parameters for that period (c.f.~Table~\ref{tab:Table4}), with the times when $Y_t=1$ indicated in gold; the degree-5 polynomial maps $\Phi_0(x)$ and $\phi_1(x)$; the resulting prices $S_t = \Phi(X_t)$.}
\label{fig:RegimeSwitchingFitB}
\end{figure}

\subsubsection{A double-Jacobi model}
Although the regime-switching model considered above gives and improved fit and a reduced volatility in comparison with the one-factor model, the volatility for the period 1st January, 2010 to 1st January, 2014 remains high, and so we consider the use of a second Jacobi process for $Y_t$ in \eqref{eq:TwoFactorS}. This fits into the framework described in Example~\ref{eg:DoubleJacobi}, if we take $B_{12}=B_{21}=0$.

In this setting, the conditional transition density $p_{Y_{m}|Y_{m-1}}(y;y_{m-1})$, and the corresponding unconditional density, are also specified using \eqref{eq:JacobiTransition}. The one-dimensional integrals described in Section~\ref{sec:Filtering} now need to be computed, and in the experiments reported below we used Gauss-Legendre integration for these computations. In this setting, we express the functions $p_{X_{m},Y_m|\vec{S}_m}(x,y;\vec{s}_m)$ as a truncated double sum expansion in Jacobi polynomials, and as before the optimal Bayes filter iteration generates a matrix equation to update the coefficients of this expansion, and the resulting normalization constants $Z_m$ generate the log likelihood $LL$ associated with the given set of parameters.

The results from applying the model to the period 1st January, 2010 to 1st January, 2014 are shown in Table~\ref{tab:Table5}, and illustrated in Figure~\ref{fig:JJA}. It can be seen from Table~\ref{tab:Table5} that the most favourable BIC score is obtained for degree~5 polynomial maps, and so this model is what is illustrated in Figure~\ref{fig:JJA}. There we see a simulation of $X_t$ and of $Y_t$, the polynomial maps, and the resulting price simulation. It is evident that the two processes can be thought of as \emph{fast} and \emph{slow}, with the (slow) $Y_t$ wandering in a somewhat leisurely fashion across the interval $[0,1]$, bringing about periods of rapid swings in prices when it is near $1$, and periods of lower prices when it is closer to $0$.

\begin{figure}
\centering
\includegraphics[width=0.85\linewidth]{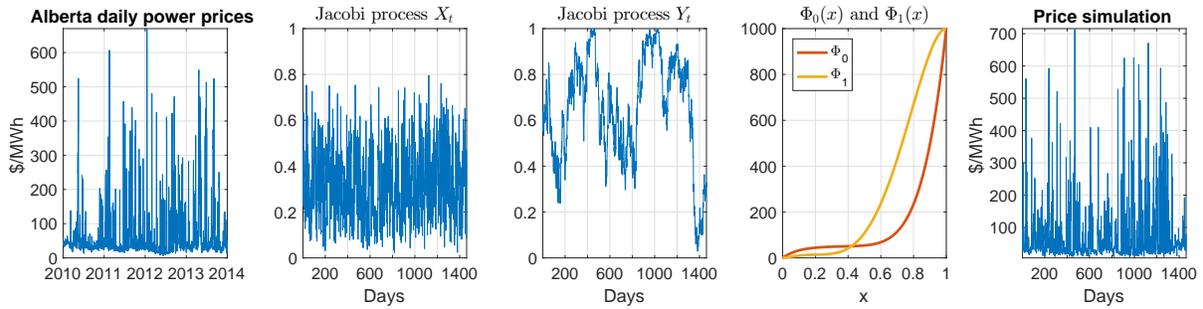}
\caption{A double Jacobi two-factor polynomial model fitted to Alberta prices. From left to right: historical daily average prices for the period 1st January, 2010 to 1st January, 2014; a simulation of $X_t$ produced using calibrated parameters for that period (c.f.~Table~\ref{tab:Table5}); a simulation of $Y_t$ produced using calibrated parameters for that period; the degree-5 polynomial maps $\Phi_0(x)$ and $\phi_1(x)$; the resulting prices $S_t = \Phi(X_t)$.}
\label{fig:JJA}
\end{figure}
\begin{table}
\begin{center}
\begin{tabular}{|r||r|r|r|r|r|}
\hline deg $\Phi$ & 2 & 3 & 4 & 5 & 6\\ 
\hline \hline $a_X$ & 2.54 & 2.52 & 2.83 & 3.10 & 3.03\\ 
\hline $b_X$ & 10.72 & 4.33 & 4.03 & 5.54 & 5.23\\ 
\hline $\sigma_X$ & 5.70 & 7.16 & 7.09 & 6.38 & 6.53\\ 
\hline \hline \hline $a_Y$ & 1.03 & 7.27 & 10.40 & 2.04 & 2.20\\ 
\hline $b_Y$ & 4.08 & 2.72 & 4.34 & 1.22 & 1.36\\ 
\hline $\sigma_Y$ & 0.95 & 1.06 & 1.02 & 1.38 & 1.36\\ 
\hline \hline $c^0_1$ &0.81 &1.00 &1.00 &0.72 &0.76\\ 
\hline $c^0_2$ & &0.79 &0.81 &0.89 &0.89\\ 
\hline $c^0_3$ & & &0.05 &0.77 &0.77\\ 
\hline $c^0_4$ & & & &0.61 &0.57\\ 
\hline $c^0_5$ & & & & &-0.08\\ 
\hline \hline $c^1_1$ &1.00 &0.99 &0.97 &0.99 &0.99\\ 
\hline $c^1_2$ & &0.61 &0.66 &0.66 &0.66\\ 
\hline $c^1_3$ & & &0.71 &-1.00 &-1.00\\ 
\hline $c^1_4$ & & & &-0.79 &-0.59\\ 
\hline $c^1_5$ & & & & &0.34\\ 
\hline \hline LL & 2915.7 & 3633.5 & 3637.1 & 3648.8 & 3650.5\\ 
\hline BIC & -5773.1 & -7194.1 & -7186.7 & -7195.7 & -7184.3\\ 
\hline
\end{tabular}

\caption{Optimal parameters, log-likelihoods and Bayesian Information Criterion (BIC) scores for calibration of the double Jacobi two-factor model for the period from 1st January, 2010 to 1st January, 2014, with various degrees of polynomial map $\Phi$. The values $a$ and $b$ are as defined in \eqref{eq:ab}. Notice that $a\geq 1$ and $b\geq 1$ in all cases. The most favourable BIC score is for degree~5.}
\label{tab:Table5}
\end{center}
\end{table}

\section{Concluding remarks}
Although in this paper we have limited our attention to polynomial diffusion models with bounded state spaces suitable for markets where prices are constrained to lie within an interval, as mentioned in the introduction the class of polynomial processes is much wider. For example, in other settings, where perhaps prices are constrained to lie within a semi-infinite interval, polynomial maps that are increasing on the half line can be used in conjunction with exponential L\'evy, CIR or (I)GBM processes. This would allow for the development of reduced form and structural models for energy markets in a range of settings, with a built-in mechanism for generating rich dynamics.

While the richness of the dynamics that can be generated with polynomial processes provides strong motivation for their use, and has been the focus of the work presented here, the characterizing feature of polynomial processes provides an even stronger motivation. That is, of course, the fact that conditional expectations and therefore forward prices have an (almost) explicit representation. This offers the potential for models that are able to capture the joint dynamics of spot prices and forward curves, and will be the focus of future work. 

\section*{Acknowledgements}
This work was supported by the Natural Sciences and Engineering Research Council of Canada (NSERC). The research leading to these results has also received funding from the European Research Council under the European Union's Seventh Framework Programme (FP/2007-2013) / ERC Grant Agreement n. 307465-POLYTE.

\section*{Appendices}
\appendix
\section{Jacobi processes}
(See \cite{DelbaenShirakawa2002}, \cite{FormanSorensen2008} \cite{GourierouxRenaultValery2007}, and \cite{Szego1959} for examples of the use of Jacobi processes in finance, and for basic information about their properties, and the properties of Jacobi polynomials.)

We consider the Jacobi process defined by \eqref{eq:Jacobi}.
If we define 
\begin{equation}
\label{eq:ab}
a = \frac{2\kappa\theta}{\sigma^2}\word{and}
b=\frac{2\kappa(1-\theta)}{\sigma^2},
\end{equation}
then we can rewrite \eqref{eq:Jacobi} in the form
\begin{equation}
\label{eq:JacobiAB}
\dd X_t = \frac{\sigma^2}{2}\big(a(1-X_t) - bX_t\big)\dd t + \sigma\sqrt{X_t(1-X_t)}\,dW_t.
\end{equation}
The structure of the process is perhaps clearer in this form. The process lives in the interval $[0,1]$. Moreover, if $a\geq 1$ (resp.~$b\geq 1$), then the boundary at $0$ (resp.~$1$) is unattainable (for a proof see, for example, \cite{DelbaenShirakawa2002}).

The transition density (from $t$ to $T>t$) is given by
\begin{equation}
\label{eq:JacobiTransition}
 p(x,T;y,t) =\sum_{n=0}^{\infty}k_n w(x)\psi_n(x)\psi_n(y)\ee^{-\mu_n(T-t)}, 
\end{equation} 
where
$\psi_n(x) = J_n(x;a+b-1,a)$,
\[ \mu_n = \kappa n+\frac{\sigma^2}{2}n(n-1) = \frac{n\sigma^2}{2}(a+b-1+n),\] 
\[ k_n = \frac{(a+b-1+2n)(a)_n(a+b)_{n}}{n!(a+b-1+n)(b)_n}\]  
and 
\[ w(x) = w(x;a,b) =\frac{\Gamma(a+b)}{\Gamma(a)\Gamma(b)} x^{a-1}(1-x)^{b-1}.\]
Note that we have used the notation $(x)_k:=\Gamma(x+k)/\Gamma(x)$.

The Jacobi polynomials $J_n(\cdot;u,v)$ satisfy the recursion (suppressing the dependence on $u$ and $v$ for convenience)
\[ x J_n(x) = \alpha_nJ_{n-1}(x) + \beta_n J_n(x) + \gamma_n J_{n+1}(x), \]
where 
\[ \alpha_n= \frac{n(v-u-n)}{(u+2n)(u+2n-1)}, \; \beta_n= \frac{2n(u+n)+v(u-1)}{(u+2n-1)(u+2n+1)}\;\text{and}\;
\gamma_n= -\frac{(v+n)(u+n)}{(u+2n+1)(u+2n)},\]
and we have $J_0(x) = 1$ and $J_1(x) = 1-\frac{(u+1)}{v}x$.

The Jacobi polynomials $J_n$ are eigenfunctions of the differential operator ${\cal L}$ defined by its action on a smooth function $g$ by
\[ \big[{\cal L}g\big](x) = \big(v-(u+1)x\big)\pdiff{g}{x}(x) + x(1-x)\pdifft{g}{x}(x). \]
The corresponding eigenvalues are $-n(u+n)$.

It follows that the functions $\psi_n(\cdot;a,b)$ are eigenfunctions of the infinitesimal generator of the Jacobi process, with eigenvalues $-\mu_n$. They are orthonormal with respect to the weight $w$, so \eqref{eq:JacobiTransition} implies that, for each $n$,
\[ \E[\psi_n(X_T) | X_t = x] = \ee^{-\mu_n(T-t)}\psi_n(x). \]

\section{Increasing polynomial maps on a bounded interval}
\label{AppendixB}
Increasing polynomial maps from $[0,1]$ to $[0,1]$ can be constructed from non-negative polynomials on $[0,1]$. Such polynomials have been characterized in various ways - going as far back as \cite{KarlinShapley1949}. Here we give an alternative characterization in terms of products of quadratic factors. The following result characterizes the set of nonnegative quadratic polynomials on $[0,1]$ that integrate to $1$.
\begin{proposition}\label{prop:icecream}
For $\alpha\in[-3,3/2]$, let $\overline{\beta}$ be defined by
\begin{equation}
\overline{\beta}(\alpha) = \begin{cases}
\frac12 + \frac{\alpha}{6}, & \alpha\in[-3,\frac53]\\
\sqrt{\frac38-(\alpha-\frac34)^2}, & \alpha\in(\frac53,\frac32].
\end{cases}
\end{equation}
Then, for any $\alpha\in[-3,\frac32]$ and $\beta$ satisfying $\abs{\beta}\leq \overline{\beta}(\alpha)$, the polynomial $q_{\alpha,\beta}(x) = \alpha x^2 + 2\beta x + 1-\frac23 \alpha$ is nonnegative on $[-1,1]$ and integrates to $2$, and $\phi(x) = q_{\alpha,\beta}(2x-1)$ is non-negative on $[0,1]$ and integrates to $1$. Moreover, any quadratic polynomial $\phi$ with these properties can be written in this way.
\end{proposition}

\begin{proof}
The proof of Proposition~\ref{prop:icecream} involves only elementary considerations. Clearly, any quadratic polynomial $q$ satisfying $\int_{-1}^1q(x)\dd x = 2$ can be written in the form $q(x) = q_{\alpha,\beta}(x) = \alpha x^2 + 2\beta x + 1-\frac23 \alpha$. In order to guarantee non-negativity, we note that $q_{\alpha,\beta}(\pm 1)\geq 0$ iff.~$1+\frac13 \alpha\geq 2\abs{\beta}$ (this region can be seen as the cone emanating from $(-3,0)$ in Figure~\ref{fig:icecream}). If there are no roots in $[-1,1]$, this condition is necessary and sufficient. There will be roots in $[-1,1]$ iff.~$\abs{\beta}<\alpha$. In this case, non-negativity holds iff.~the minimum value is non-negative, and this results in the condition that $(\alpha,\beta)$ be inside the ellipsoidal region shown in Figure~\ref{fig:icecream}. 
\end{proof}
All increasing polynomial maps $\Phi: [0,1]\mapsto[0,1]$ can be constructed from pairs of parameters $(\alpha_1,\beta_1),\dots,(\alpha_k,\beta_k)$ by writing
\begin{equation}
\label{eq:phiproduct}
\phi(x) = \Pi_{i=1}^kq_{\alpha_i,\beta_i}(2x-1),
\end{equation}
and setting 
\begin{equation}
\Phi(x) = \frac{\int_0^x\phi(t)\dd t}{\int_0^1\phi(t)\dd t}.
\label{eq:Phi}
\end{equation}
Note that the degree of $\Phi$ is $2k+1$, if $\alpha_i\ne 0$ for each $i$. Note too that, if $\phi(x)\geq 0$ on $[0,1]$, any real zeros of $\phi$ in $[0,1]$ have to occur in pairs, so that the Fundamental Theorem of Algebra implies that any nonnegative polynomial on $[0,1]$ can be written in the form \eqref{eq:phiproduct}.

\begin{figure}
\includegraphics[width=\linewidth]{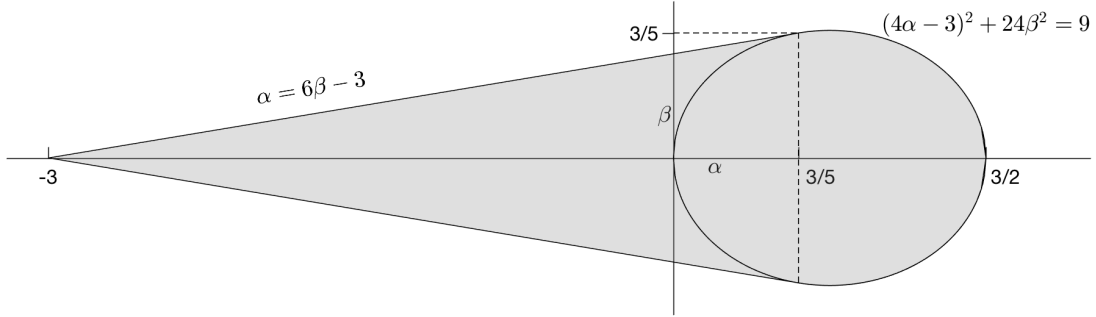}
\caption{The set of parameters $(\alpha,\beta)$ that generate normalize non-negative quadratic polynomials on $[-1,1]$ (see Proposition~\ref{prop:icecream}).}
\label{fig:icecream}
\end{figure}

% ----------------------------------------------------------------
\bibliographystyle{plain}
\bibliography{Filipovic_Larsson_Ware}

\end{document}